\documentclass[journal]{article}

\usepackage{latexsym,amssymb}
\usepackage{enumerate, verbatim}
\usepackage{graphics}
\usepackage[latin1]{inputenc}
\usepackage[T1]{fontenc}
\usepackage{amsmath, amsfonts, amsthm, amscd}
\usepackage{amssymb}

\newtheorem{teo}{Theorem}
\newtheorem{prop}{Proposition}
\newtheorem{lemma}{Lemma}
\newtheorem{coro}{Corollary}
\newtheorem{remark}{Remark}
\newtheorem{defi}{Definition}

\newcommand{\N} {\ensuremath{\mathbb{N}}}

\newcommand{\F} {\ensuremath{\mathbb{F}}}

\newcommand{\Aut} {\textnormal{\textrm{Aut}}}
\newcommand{\soc} {\textnormal{\textrm{soc}}}

\newcommand{\ho} {\textnormal{\textrm{Hom}}}
\begin{document}

\title{Automorphisms of order $2p$ in binary self-dual extremal codes of length a multiple of $24$}

\author{Martino~Borello and Wolfgang~Willems
\thanks{M. Borello is with the Dipartimento di Matematica e
Applicazioni, Universit\`{a} degli Studi di Milano Bicocca, 20125
Milan, Italy, e-mail: m.borello1@campus.unimib.it.}
\thanks{W. Willems is with the Institut f\"{u}r Algebra und Geometrie,
Fakult\"{a}t f\"{u}r Mathematik, Otto-von-Guericke Universit\"{a}t,
39016 Magdeburg, Germany, e-mail: willems@ovgu.de.}}

\maketitle

\begin{abstract}
Let $C$ be a binary self-dual code with an automorphism $g$ of order
$2p$, where $p$ is an odd prime, such that $g^p$ is a fixed point
free involution. If $C$ is extremal of length a multiple of $24$ all
the involutions are fixed point free, except
 the Golay Code and eventually   putative codes of length
$120$.

Connecting module theoretical properties of a self-dual code $C$
with coding theoretical ones of the subcode $C(g^p)$ which consists
of the set of fixed points of $g^p$, we prove that $C$ is a
projective $\F_2\langle g \rangle$-module if and only if a natural
projection of $C(g^p)$ is a self-dual code. We then discuss easy to
handle criteria to decide if $C$ is projective or not.

  As an application we consider in the last part extremal self-dual codes
of length $120$, proving that their automorphism group does not
contain elements of order $38$ and $58$.
\end{abstract}

\section{Introduction}

Binary self-dual extremal codes of length a multiple of $24$ are
binary self-dual codes with parameters $[24m,12m,4m+4]$. They are
interesting for various algebraic and geometric reasons; for
example, they are doubly even \cite{articolorains} and all codewords
of a fixed given nontrivial weight support a $5$-design
\cite{articolo10}. Very little is known about this family of codes:
for $m=1$ we have the Golay Code $\mathcal{G}_{24}$ and for $m=2$
there is the extended quadratic residue code $XQR_{48}$, but no
other examples are known so far.

A classical way of approaching the study of such codes is through
the inve\-stigation of their automorphism group. In this paper we
focus our attention to  automorphisms of order $2p$, where $p$ is an
odd prime. There are elements of this type in the automorphism group
of $\mathcal{G}_{24}$ and  $XQR_{48}$, while it was recently proved
\cite{articolomio} that for $m=3$ no automorphisms of order $2p$
occur. The problem is totally open for $m>3$. It is known
\cite{articolo2} that for $m\not\in \{1,5\}$ the involutions are
fixed point free. So we will restrict our study to those
automorphisms $g$ of order $2p$ whose $p$-power acts fixed point
freely.

In the first part of the paper we connect module theoretical
properties of a self-dual code $C$ with coding theoretical ones of
the subcode $C(g^p)$ which  consists of the fixed points of $g^p$.
More precisely, we prove in Theorem \ref{proj} that $C$ is a
projective $\F_2\langle g \rangle$-module if and only if a natural
projection of $C(g^p)$ is a self-dual code. In the second part, i.e.
section \ref{120}, we  apply  these results to the case $m=5$. In
particular we prove that there are no automorphisms of order $2\cdot
19$ and $2 \cdot 29$.
 All computations of the last part are carried out with \textsc{Magma} \cite{articolo19}.

\section{Preliminaries}

From now on a code always means a binary linear code and $K$ always denotes the field $\F_2$ with two elements. \\
Let $C$ be a code and let $g\in\Aut(C)$. We denote by
$$C(g)=\{c\in C \ | \ c^g=c\}$$
the subcode of $C$ consisting of all codewords which are fixed by
$g$. It is easy to see that a codeword $c=(c_1,\ldots,c_n)$ is fixed
by $g$ if and only if $c_i=c_{i^g}$ for every $i\in\{1,\ldots,n\}$,
i.e., if and only if $c$ is constant on the orbits of $g$.

\begin{defi} For an odd prime $p$ let $s(p)$ denote the smallest $s\in \N$ such that $ p \mid 2^s -1$. Note that $s(p)$
 is the multiplicative order of $2$ in $\F_p^\ast$.
\end{defi}

The next lemma is a well-known fact in modular representation
theory. For those who are not familiar with representation theory we
recall here some of the notions we need. Let $G$ be a group. A
projective indecomposable $KG$-module is a direct summand W of the
group algebra $KG$ which cannot be written as $W=W' \oplus W''$ with
$KG$-modules $W' \not=0 \not= W''$.  Such a module $W$ has a unique irreducible
submodule, say $V$, called the socle of $W$, and a unique
irreducible factor module which is isomorphic to $V$. We call $W$
which is (up to isomorphism) uniquely determined by $V$
  the projective cover of $V$.  Projective covers for irreducible modules always exist (actually they exist for any
    finite dimensional $KG$-module).
 For these facts and some basics in modular representation theory (and only those are needed in
this article) the reader is referred to chapter VII of \cite{HB}. Finally note that the action of $G$ on a module is always from the right
in this article.

\begin{lemma}\label{basiclemma}
Let $\nu=\frac{p-1}{s(p)}$, where $p$ is an odd prime,  and let
$G=\langle g \rangle$, a cyclic group of order $2p$. Then we have.

\begin{itemize}
  \item[\rm a)] There are $1+\nu$ irreducible $KG$-modules $V_0,V_1,\ldots,V_\nu$, where $V_0=K$ {\rm(}the trivial module{\rm)} and $\dim V_i=s(p)$ for $i\in\{1,\ldots,\nu\}$.
  \item[\rm b)] For $i=0,\ldots,\nu$ the projective indecomposable cover $W_i$ of $V_i$ is a nonsplit extension
  $W_i=\begin{array}{c} V_i \\ V_i\end{array}$  of $V_i$ by $V_i$ Furthermore,  $$KG=W_0\oplus W_1 \oplus \ldots \oplus W_\nu.$$
\end{itemize}
\end{lemma}

In order to understand codes with automorphisms of order $2p$ we
need the following result on self-dual modules which improves Proposition 3.1 of
\cite{MW}. Recall that a $KG$-module $V$ is self-dual if $V \cong V^*$ (as $KG$-modules). Here
$g \in G$ acts on $V^* = \ho_K(V,K)$ by $$ fg(v) = f(vg^{-1})$$ $ where\, f \in V^*, g \in G $ and $ v \in V$.

\begin{prop} \label{self-dual} Let $G=\langle g\rangle$ be a cyclic group of odd prime order $p$.
\begin{itemize}
 \item[\rm a)]If $s(p)$ is even,
              then all ir\-re\-ducible $KG$-modules are self-dual.
 \item[\rm b)] If $s(p)$ is odd, then the trivial module is the only self-dual irreducible $KG$-module.
\end{itemize}
\end{prop}

\begin{proof} a) Let $s(p)=2t$ and let $E=\F_{2^{2t}}$ be the extension field of $K=\F_2$ of degree $2t$.
Furthermore, let $W$ be an irreducible nontrivial $KG$-module. In
particular, $W$ has dimension $2t$. By Theorem 1.18 and Lemma 1.15
in Chap. VII of \cite{HB}, we have
\begin{equation} \label{extension} W \otimes_K E = \oplus_{\alpha \in {\rm Gal}(E/K)} V^\alpha  \end{equation}
where $V$ is an irreducible $EG$-module and $V^\alpha$ is the
$\alpha$-conjugate module of $V$. The action of $g \in G$ on $V^\alpha$ is given by the matrix $(a_{ij}(g)^\alpha)$ if
$g$ acts via the matrix $(a_{ij}(g))$ on $V$.
Since $p \mid (2^t+1)(2^t-1)$ we
get $ p \mid 2^t+1$. Clearly, the Galois group ${\rm Gal(E/K)}$ of $E$ over $K$ (i.e. the group of field automorphisms of $E$ which leave
the subfield $K$ elementwise fixed) consists of all
automorphisms of the form $ x \mapsto x^{2^k}$ where $ 0 \leq k \leq
2t-1$ (see \cite{HP}, section 3.6).

If $V=\langle v \rangle$  then $vg = \epsilon v$ where $\epsilon$ is
a nontrivial $p$-th root of unity in $E$. Since $p \mid 2^t+1$ we
obtain $ \epsilon^{2^t+1} = 1$, hence $\epsilon^{2^t} =
\epsilon^{-1}$. Thus there is an $\alpha \in {\rm Gal}(E/K)$ such
that
$$   V^* \cong V^\alpha$$
and equation (\ref{extension}) implies $W \cong W^*$. \\
b) Now let $s(p)= t$ be odd. As above the irreducible module $W$ is
self-dual if and only if $V^* \cong V^\alpha$ for some $\alpha  \in
{\rm Gal}(\F_{2^t}/K)$, or equivalently if and only if
$\epsilon^\alpha = \epsilon^{-1}$. Suppose that such an $\alpha$
exists. Then we may write $\epsilon^\alpha = \epsilon^{2^k}$ where
$0 \leq k \leq t-1$. Hence $\epsilon^{2^k} = \epsilon^{-1}$ for some
$0 \leq k \leq t-1$ and therefore $ 2^k \equiv -1 \bmod p $. Now
$2^{2k} \equiv 1 \bmod p$ forces $ t \mid 2k$. Since $t$ is odd we
get $ t \mid k \leq t-1$, a contradiction.
\end{proof}

\begin{remark} \label{remark1} {\rm According to Lemma 3.5 in \cite{MW} we have $s(p)$ even if $p \equiv \pm 3 \bmod 8$ and $s(p)$ odd if
$p \equiv -1 \bmod 8$.}
\end{remark}

\begin{remark} {\rm Since $KG\cong KG^\ast$ (see \cite{HB}, Chap. VII, Lemma 8.23),  Lemma \ref{basiclemma} and Proposition \ref{self-dual} imply the following.
\begin{itemize}
 \item[\rm a)]If $s(p)$ is even, then
              $$KG=W_0\oplus W_1\oplus\ldots\oplus W_\nu$$
  with $W_i\cong W^\ast_i$ for all $i\in\{0,\ldots,\nu\}$.
 \item[\rm b)] If $s(p)$ is odd, then $\nu$ is even (put $\nu=2t$) and
           $$KG=W_0\oplus W_1\oplus\ldots\oplus W_{2t}$$
  with $W_0\cong W_0^\ast$ and $W_i\cong W^\ast_{2i}$ for all $i\in\{1,\ldots,t\}$.
\end{itemize}
}
\end{remark}

\section{Automorphisms of order $2p$ in self-dual codes}

Throughout this section let $C$ be a self-dual code of length $n$. In particular $n$ is even. Suppose that $g\in \Aut(C)$ is of order $2p$, where $p$ is an odd prime. Furthermore suppose that the involution $h=g^p$ acts fix point freely on the $n$ coordinates. Without loss of generality, we may assume that $h=g^p=(1,2)(3,4)\ldots(n-1,n)$.\\

We consider the maps $\pi=\pi_2:C(h)\rightarrow K^{\frac n 2}$,
where $$(c_1,c_1,c_2,c_2,\ldots,c_{\frac n 2},c_{\frac n
2})\overset{\pi}{\mapsto}(c_1,c_2,\ldots,c_{\frac n 2}),$$ and
$\phi:C\rightarrow K^{\frac n 2}$, where
$$(c_1,c_2,\ldots,c_{n-1},c_n)\overset{\phi}{\mapsto}(c_1+c_2,\ldots,c_{n-1}+c_n).$$
According to Theorem 1 of \cite{Baut2} we have
$$\phi(C)\subseteq \pi(C(h))=\phi(C)^\perp.$$
In particular,
$$  \phi(C) = \pi(C(h)) = \phi(C)^\perp \quad \mbox{(i.e. $\pi(C(h))$ is self-dual)} $$
if and only if
$$ \dim \, \pi(C(h)) = \dim \, C(h) = \frac{n}{4}.$$

To state one of the main results recall that a projective $KG$-module
is a finite direct sum of projective indecomposable modules, or in other words, it is a direct summand of a finite
direct sum of copies isomorphic to the group algebra $KG$ (as $KG$-modules).

\begin{teo}\label{proj}
The code $C$ is a projective $K\langle g \rangle$-module if and only
if $\pi(C(h))$ is a self-dual code.
\end{teo}

\begin{proof} First note that for an arbitrary finite group $G$
a $KG$-module is projective if and only if its restriction to a
Sylow $2$-subgroup is projective (\cite{HB}, Chap. VII, Theorem
7.14). Thus we have to consider the restriction $C_{|_{\langle
h\rangle}}$, i.e., $C$ with the action of $\langle h \rangle$. As a
$K\langle h \rangle$-module we may write
$$C\cong \underbrace{R\oplus\ldots\oplus R}_{a \ \text{times}}\oplus \underbrace{K\oplus \ldots \oplus K}_{\frac n 2 -2a \ \text{times}},$$
where $R$ is the regular $K\langle h \rangle $-module and $K$ is the
trivial one. If $\soc(C)$ denotes the socle of $C$, i.e. the largest
completely reducible $K\langle h \rangle $-submodule of $C$, then
$$C(h)=\soc(C)=\underbrace{K\oplus\ldots\oplus K}_{a \ \text{times}}\oplus \underbrace{K\oplus \ldots \oplus K}_{\frac n 2 -2a \ \text{times}}\cong K^{\frac n 2 -a}.$$ Thus
$C$ is projective if and only if $\frac n 2 -2a=0$, hence if and
only if $a=\frac n 4$. This happens if and only if $\dim C(h)=\frac
n 4$. This is equivalent to the fact that $\pi(C(h))$ is self-dual.
\end{proof}

\begin{remark} {\rm
If $n\equiv 2 \bmod 4$, then $\pi(C(h))\subseteq K^{\frac{n}{2}}$
cannot be self-dual, since $\frac{n}{2}$ is odd.}
\end{remark}

\begin{remark} {\rm
In $\mathcal{G}_{24}$ and $XQR_{48}$ the subcodes fixed by fixed
point free acting
 involutions have self-dual projections. Thus we
wonder if this holds true for all extremal self-dual codes of length
a multiple of $24$.}
\end{remark}

Next we deduce  some properties of $C$ related to the action of the
automorphism $g$ of order $2p$.  This may help to decide whether
$\pi(C(h))$ is self-dual or not. For  completeness we
 treat both cases $n\equiv 2 \bmod 4$ and $n\equiv 0 \bmod 4$.

Since $h$ acts fixed point freely, $g$ has $x$ $2p$-cycles and $w$
$2$-cycles, with
\begin{equation}\label{eq1}
n=2px+2w.
\end{equation}

\noindent Thus, as a $K\langle g \rangle$-module, we have the
decomposition $$K^n=\underbrace{K\langle g \rangle\oplus
\ldots\oplus K\langle g \rangle}_{x \ \text{times}}\oplus
\underbrace{K\langle h\rangle \oplus \ldots \oplus K\langle
h\rangle}_{w \ \text{times}}.$$

\noindent Using Lemma \ref{basiclemma} and  $V_0 \cong K$, we get

$$K^n=
\underbrace{\begin{array}{c} V_0 \\ V_0 \end{array}\oplus \ldots
\oplus \begin{array}{c} V_0 \\ V_0 \end{array}}_{x + w \
\text{times}} \oplus \ldots \oplus \underbrace{\begin{array}{c}
V_{\nu} \\ V_{\nu} \end{array}\oplus \ldots \oplus \begin{array}{c}
V_{\nu} \\ V_{\nu} \end{array}}_{x \ \text{times}}.
$$

\noindent
The action of $\langle g \rangle$ on $K^n$ and the self-duality of $C$ restrict the possibilities for $C$ as a subspace of $K^n$. \\

More precisely, we have

\begin{prop}\label{decomp}
As a $K\langle g \rangle$-module, the code $C$  has the following
structure.
$$C=\underbrace{\begin{array}{c} V_0 \\ V_0 \end{array}\oplus \ldots \oplus \begin{array}{c} V_0 \\ V_0 \end{array}}_{y_0 \ \text{times}}\oplus\underbrace{\begin{array}{c}          V_0 \end{array}\oplus \ldots \oplus \begin{array}{c}         V_0 \end{array}}_{z_0 \ \text{times}}\oplus \ldots $$ $$ \mbox{} \qquad \ldots \oplus\underbrace{\begin{array}{c} V_{\nu} \\ V_{\nu} \end{array}\oplus \ldots \oplus \begin{array}{c} V_{\nu} \\ V_{\nu} \end{array}}_{y_{\nu} \ \text{times}}
 \oplus \underbrace{\begin{array}{c}     V_{\nu} \end{array}\oplus \ldots \oplus \begin{array}{c}   V_{\nu} \end{array}}_{z_{\nu}},
$$
\noindent where
\begin{itemize}
  \item[\rm 1)] $2y_0+z_0=x +w$,
  \item[\rm 2a)] $2y_i+z_i=x$ for all $i\in\{1,\ldots,\nu \}$, if $s(p)$ is even,
  \item[\rm 2b)] $z_i=z_{2i}$ and $y_{i}+y_{2i} + z_i=x$ for all $i\in\{1,\ldots,t\}$, if $s(p)$ is
  odd.
\end{itemize}
\end{prop}

\begin{proof}
Since $C=C^\perp$ we see by a proof similar to that of Proposition
2.3 in \cite{W} that $K^n/C \cong C^*$. The conditions on the
multiplicities are an easy consequence of this fact. Let us prove,
for example, part 2b): if
$$C=\ldots\oplus \underbrace{\begin{array}{c} V_i \\ V_i \end{array}\oplus \ldots \oplus \begin{array}{c} V_i \\ V_i \end{array}}_{y_i \ \text{times}}\oplus\underbrace{\begin{array}{c} V_i \end{array}\oplus \ldots \oplus \begin{array}{c} V_i \end{array}}_{z_i \ \text{times}}\oplus \ldots $$ $$ \mbox{} \qquad \ldots \oplus\underbrace{\begin{array}{c} V_{2i} \\ V_{2i} \end{array}\oplus \ldots \oplus \begin{array}{c} V_{2i} \\ V_{2i} \end{array}}_{y_{2i} \ \text{times}}
 \oplus \underbrace{\begin{array}{c}     V_{2i} \end{array}\oplus \ldots \oplus \begin{array}{c}   V_{2i} \end{array}}_{z_{2i}}\oplus\ldots,
$$
then
$$K^n/C=\ldots\oplus \underbrace{\begin{array}{c} V_i \\ V_i \end{array}\oplus \ldots \oplus \begin{array}{c} V_i \\ V_i \end{array}}_{x-z_i-y_i \ \text{times}}\oplus\underbrace{\begin{array}{c} V_i \end{array}\oplus \ldots \oplus \begin{array}{c} V_i \end{array}}_{z_i \ \text{times}}\oplus \ldots $$ $$ \mbox{} \qquad \ldots \oplus\underbrace{\begin{array}{c} V_{2i} \\ V_{2i} \end{array}\oplus \ldots \oplus \begin{array}{c} V_{2i} \\ V_{2i} \end{array}}_{x-z_{2i}-y_{2i} \ \text{times}}
 \oplus \underbrace{\begin{array}{c}     V_{2i} \end{array}\oplus \ldots \oplus \begin{array}{c}   V_{2i}
 \end{array}}_{z_{2i}}\oplus\ldots
$$
and since $V_i\cong V_{2i}^\ast$,
$$C^\ast=\ldots\oplus \underbrace{\begin{array}{c} V_{2i} \\ V_{2i} \end{array}\oplus \ldots \oplus \begin{array}{c} V_{2i} \\ V_{2i} \end{array}}_{y_i \ \text{times}}\oplus\underbrace{\begin{array}{c} V_{2i} \end{array}\oplus \ldots \oplus \begin{array}{c} V_{2i} \end{array}}_{z_i \ \text{times}}\oplus \ldots $$ $$ \mbox{} \qquad \ldots \oplus\underbrace{\begin{array}{c} V_{i} \\ V_{i} \end{array}\oplus \ldots \oplus \begin{array}{c} V_{i} \\ V_{i} \end{array}}_{y_{2i} \ \text{times}}
 \oplus \underbrace{\begin{array}{c}     V_{i} \end{array}\oplus \ldots \oplus \begin{array}{c}   V_{i}
 \end{array}}_{z_{2i}}\oplus\ldots.
$$
Thus $z_i=z_{2i}$ and $x-z_i-y_i=y_{2i}$.
\end{proof}

\noindent  Proposition \ref{decomp} implies that
\begin{equation}\label{eqfix}
\phi(C)^\perp=\pi(C(h))=\pi\left(\bigoplus_{i=0}^{\nu}\underbrace{V_i
\oplus \ldots \oplus  V_i }_{y_i+z_i \ \text{times}}\right).
\end{equation}

\noindent Since $\ker \phi = C(h)$, we furthermore have
$$\phi(C)\cong C/\ker \phi\cong \bigoplus_{i=0}^{\nu}\underbrace{V_i\oplus \ldots\oplus V_i}_{y_i \ \text{times}},
$$
\noindent which leads to

$$
\phi(C)^\perp/\phi(C)\cong
\bigoplus_{i=0}^{\nu}\underbrace{V_i\oplus \ldots\oplus V_i}_{z_i \
\text{times}}.
$$

\noindent Taking dimensions we get
\begin{equation} \label{xxx}
\dim \phi(C)^\perp/\phi(C) = z_0 +
s(p)\left(\sum_{i=1}^{\nu}z_i\right).
\end{equation}

\begin{prop} \label{lemma2}
With  the notations used in Proposition \ref{decomp} we have
\begin{itemize}
  \item[\rm a)] $ x \equiv w  \bmod 2$, if $n\equiv 0 \bmod 4$,
  \item[\rm b)] $x \not\equiv w  \bmod 2$, if $n\equiv 2 \bmod 4$.
\end{itemize}
Furthermore, if $s(p)$ is even, then
$$x\equiv z_1 \equiv \ldots \equiv z_\nu \bmod 2.$$
\end{prop}

\begin{proof}
a) and b) follow immediately from \eqref{eq1}. The last fact is a
consequence of $2y_i + z_i =x$, if $s(p)$ is even, which is stated
in  Proposition \ref{decomp}.
\end{proof}

\begin{coro} \mbox{}
\begin{itemize}
  \item[\rm a)] $\phi(C)^\perp/\phi(C) $ is of even dimension, if $n\equiv 0 \bmod 4$,
  \item[\rm b)] $\phi(C)^\perp/\phi(C) $ is of odd dimension, if $n\equiv 2 \bmod 4$.
\end{itemize}
\end{coro}

\begin{proof}
First note that  $s(p)\sum_{i=1}^\nu z_i\equiv 0 \bmod 2$ whatever
the parity of $s(p)$ is. In case $s(p)$ odd this follows from $ z_i
= z_{2i}$ for $i\in\{1, \ldots, 2t=\nu\}$ (see Proposition
\ref{decomp}). Furthermore, $z_0 \equiv x+w \bmod 2$, hence $z_0$
even, if $4 \mid n$, and $z_0$ odd, if $n \equiv 2 \bmod 4$,
according to Proposition \ref{lemma2}. Thus
 \eqref{xxx} yields
 $$\dim \phi(C)^\perp/\phi(C) \equiv z_0 \equiv 0 \bmod 2, \  \mbox{if} \  n\equiv 0 \bmod 4$$ and
 $$ \dim \phi(C)^\perp/\phi(C) \equiv z_0 \equiv 1 \bmod 2, \ \mbox{if} \ n\equiv 2 \bmod 4.$$
\end{proof}

\begin{coro}\label{coro2}
Let $n\equiv 0 \bmod 4$ and let $s(p)$ be even. If $w$ is odd, then
$$\dim C(h) = \dim \pi(C(h)) \geq \frac{n}{4} + \frac{s(p) \nu
}{2}=\frac{n}{4} + \frac{p-1}{2}.$$ In particular, $\phi(C) <
\phi(C)^\perp$.
\end{coro}

\begin{proof}
By Proposition \ref{lemma2}, the condition $4 \mid n$ forces that
$w$ and $x$ have the same parity. Thus $w$ odd implies that $x$ is
odd and by Proposition \ref{decomp}, we get $ z_i \geq 1 $ for $i=1,
\ldots \nu$.
Since
$$ \phi(C) \subseteq \phi(C)^\perp =\pi(C(h)) \subseteq K^{\frac{n}{2}},$$
we have
$$ \dim \pi(C(h)) \geq \frac{n}{4} + \frac{1}{2} \dim \phi(C)^\perp/\phi(C).$$
Therefore, according to (\ref{xxx}),
$$\dim C(h) = \dim \pi(C(h)) \geq \frac{n}{4} + \frac{s(p) \nu }{2}=\frac{n}{4} + \frac{p-1}{2}.$$
\end{proof}

\begin{remark} {\rm
We may ask whether  the converse of Corollary \ref{coro2} holds
true; i.e., does $\phi(C) < \phi(C)^\perp$ always implies that $w$
is odd? This is not true.   For instance, there exist self-dual
$[36,18,8]$ codes and automorphisms of order $6$ (note that $s_2(3)$
is even) for which $\pi(C(h))$ is not self-dual, but $w$ is even. }
\end{remark}

\begin{coro}\label{coro3} Let $n\equiv 0 \bmod 4$ and let $s(p)$ be even.
If $g $ has an odd number of cycles of order $2$, then $C$ is not
projective as a $K\langle g\rangle$-module.
\end{coro}

\begin{proof} If the number of $2$-cycles of $g$ is odd, then $w$ is odd. Thus, by Corollary \ref{coro2} and Theorem \ref{proj}, the assertion follows.
\end{proof}

To state further results we need the following notation about the structure of the
automorphisms.

\begin{defi}{\rm  We say that an automorphism  of prime order $p$ of a code is of type $p$-$(\alpha,\beta)$ if it has $\alpha$ $p$-cycles and $\beta$ fixed points.
Furthermore an automorphism of order $2p$ is of type
$2p$-$(\alpha,\beta,\gamma;\delta)$ if it has $\alpha$ $2$-cycles,
$\beta$ $p$-cycles, $\gamma$ $2p$-cycles and $\delta$ fixed points.
}
\end{defi}

Since $\Aut(C)\leq \mathcal{S}_n$, the largest possible prime which
may occur as the order of an automorphism of a self-dual code of
length $n$ is $p=n-1$. If $n\equiv 0 \bmod 8$, then $s(p)$ is odd (see Remark \ref{remark1}).
Obviously, in this case we cannot have an automorphism of order
$2p$.

Let $C$ be an extremal self-dual code of length $n \geq 48$.
According to Theorem 7 in \cite{BMW}  an automorphism of type
$p$-$(\alpha,\beta)$ with $p>5$ satisfies $\alpha \geq \beta$.
Hence the second largest possible prime $p$ satisfies $n=2p+2$.

\begin{coro} Let $C$ be a self-dual code of length $n=2p +2 $, where $p$ is an odd prime, and minimum distance greater than $4$. Suppose that involutions in
$\Aut(C)$ are fixed point free.  If $s(p)$ is even, then $\Aut(C)$
does not contain
an element of order $2p$. \\
In case $C$ is doubly even, the condition $s(p)$ even may be
replaced by the condition $p \not\equiv -1 \bmod 8$.
\end{coro}

\begin{proof} Suppose that $g$ is an automorphism of order $2p$. Thus $g$ has a cycle of length $2p$ and one of length $2$. As above let $h=g^p$.
By Corollary \ref{coro2}, we get
$$ \dim \, \pi(C(h)) \geq \frac{n}{4} + \frac{p-1}{2} = p.$$
Since $ \pi(C(h)) \leq K^{\frac{n}{2}} = K^{p+1}$, we see that $
\pi(C(h))$ has minimum distance $1$ or $2$, a contradiction.

In case that $C$ is doubly even we only have to show that $ p \equiv
1  \bmod 8$ does not occur (see Remark \ref{remark1}). If $ p \equiv
1 \bmod 8$ then $n=2p+2 \equiv 4 \bmod 8$, contradicting the Theorem
of Gleason (see \cite{HP}, Corollary 9.2.2).
\end{proof}

\begin{coro} \label{prime}
Let $C$ be an extremal self-dual  code of length $n=24m$.  Let $g
\in \Aut(C)$ be an element
 of type $2p$-$(w,0,x;0)$. If $s(p)$ is even and $w$ is odd, then
 $ p \leq \frac{n}{4} -1$.

\end{coro}

\begin{proof}
 By Corollary \ref{coro2},
 $\pi(C(h))$ has parameters $[12m,\geq 6m+\frac{p-1}{2},\geq 2m+2]$.
 According to the
 Griesmer bound (see \cite{HP}, Theorem 2.7.4), we have
 $$\begin{array}{rl} 12m & \geq \sum_{i=0}^{6m + \frac{p-1}{2}-1} \left \lceil \frac{2m+2}{2^i} \right \rceil
 \\
 &\geq (2m+2) + (m+1) + (6m + \frac{p-1}{2}) -2.  \end{array}$$
 This implies $ p \leq 6m-1 = \frac{n}{4} -1$.
\end{proof}

Clearly, the estimation in Corollary \ref{prime} is very crude for
$m$ large. For instance, if $m =5$  the statement in Corollary
\ref{prime} leads to $p\leq 29$, but computing all terms in the sum
shows that $ p \leq 23$.

\section{Application to extremal self-dual codes of \\length $120$}\label{120}

From now on $C$ is supposed to be a self-dual $[120,60,24]$ code.
The following (see \cite{thesis}) is the state of art about the automorphisms of $C$.\\

Automorphisms of odd prime order which may occur in $\Aut(C)$ are of
type $29$-$(4,4)$, $23$-$(5,5)$, $19$-$(6,6)$, $7$-$(17,1)$,
$5$-$(24,0)$ or  $3$-$(40,0)$.  Automorphisms of order $2$ can  only
be of type $2$-$(48,24)$ or $2$-$(60,0)$.
Automorphisms of  possible composite odd order are of type $3\cdot 5$-$(0,0,8;0)$, $3\cdot 19$-$(2,0,2;0)$ or $5\cdot 23$-$(1,0,1;0)$. \\

Thus we may ask about elements $g \in \Aut(C)$ of order $2p$ where
$p$ is an odd prime. Note that the involution $h=g^p$ has no or
exactly $24$ fixed points, by \cite{articolo2}.

\begin{lemma}\label{struct}
If the involution $h$ has no fixed points, then $g$ is of type
\begin{itemize}
  \item $2\cdot 29$-$(2,0,2;0)$,
  \item $2\cdot 19$-$(3,0,3;0)$,
  \item $2\cdot 5$-$(0,0,12;0)$,
  \item or \ $2\cdot 3$-$(0,0,20;0)$.
\end{itemize}
If $h$ has $24$ fixed points then $g$ is of type
\begin{itemize}
  \item $2\cdot 23$-$(2,1,2;1)$,
  \item or \  $2\cdot 3$-$(0,8,16;0)$.
\end{itemize}
Note that $\Aut(C)$ does not contain  elements of order $2\cdot 7$.
\end{lemma}

\begin{proof}
The proof is straightforward by considering the cycle-structures
using \cite{thesis}.
\end{proof}

The above cycle structures show that only elements of order $2 \cdot
19$  satisfy the hypothesis of Corollary \ref{coro2}. In this case
$s(19)$ is even and so we have
$$\dim C(g^{19})\geq \frac{120}{4} + \frac{19-1}{2}=39.$$
Thus  $\pi_2(C(g^{19}))$ is a $[60,\geq 39,\geq 12]$ code. According
to Grassl's list \cite{Grassl} a $[60,\geq 39]$ code has minimum
distance at most $10$. Therefore we can state the following.

\begin{prop}
The automorphism group of an extremal self-dual $[120,60,24]$ code
does not contain elements of order $38$.
\end{prop}
Next we consider automorphisms of order 58. By Lemma \ref{struct},
we know that $g$ is of type $2\cdot 29$-$(2,0,2;0)$. Therefore $g^2$
is of type $29$-$(4,4)$ and $g^{29}$ is of type $2$-$(60,0)$. Thus,
without loss of generality, we may assume that

$$g^2=(1,\ldots,29)(30,\ldots,58)(59,\ldots,87)(88,\ldots,116)$$

and

$$g^{29}=(1,30)\ldots(59,88)\ldots(117,118)(119,120).$$

If $\pi_{29}:C(g^2)\rightarrow \F_2^8$ is defined by
$$(v_1,\ldots,v_{120})\mapsto(v_1,v_{30},v_{59},v_{88},v_{117},v_{118},v_{119},v_{120})$$
then $\pi_{29}(C(g^2))$ is a self-dual $[8,4]$ code according to
\cite{articolo18}, and clearly, the minimum distance must be greater
than or equal to $4$, since $C$ is doubly-even. It is well-known
that, up to equivalence, the only code with such parameters is the
extended Hamming code $\hat{\mathcal{H}}_3$.

According to Lemma \ref{basiclemma} the structure of the ambient
space $K^{120}$, viewed as a module for the group $\langle g
\rangle$, is as follows:
$$K^{120}=\begin{array}{cccc} K & K & K & K \\ K & K & K & K  \end{array} \oplus \begin{array}{cc} V & V  \\ V & V \end{array}$$
where $\dim V = 28$. Since $C(g^2)$ has dimension $4$, the code
$C(g)=(C(g^2))(g^{29})$ has dimension at least $2$. By calculations
we verify that
$$\dim((\pi_{29}^{-1}(A))(g))\leq 2$$
for every $A\in \hat{\mathcal{H}}_3^{\mathcal{S}_8}$, which denotes
the set of all self-dual $[8,4,4]$ codes. Note that there are only a
few computations since
$|\hat{\mathcal{H}}_3^{\mathcal{S}_8}|=\frac{|\mathcal{S}_8|}{|\Aut(\hat{\mathcal{H}}_3)|}=30$.
Thus $\dim \, C(g)=2$ and there are only two possible structures for
$C$, namely

\begin{itemize}
  \item[\rm a)] \quad $C=\begin{array}{cc} K & K \\ K & K \end{array} \oplus \begin{array}{c}  V  \oplus  V \end{array}$ or
  \item[\rm b)] \quad $C=\begin{array}{cc} K & K \\ K & K \end{array} \oplus \begin{array}{c} V  \\ V \end{array}$.
\end{itemize}

Next we look at $C(g^{29})$ which may be written as
$C(g^{29})=B\otimes \langle(1,1)\rangle$, where $B
=\pi_2(C(g^{29}))$ is a $[60,\geq 30,\geq 12]$ code. In case a) we
have $\dim B = 58$, a contradiction. Thus case b) occurs. According
to Theorem \ref{proj},
 $C$ is projective and $B$ is a self-dual $[60,30,12]$ code. Furthermore $B$ has an automorphism of type $29$-$(2,2)$.

\begin{prop}
Every self-dual $[60,30,12]$ code $B$ with an automorphism of type
$29$-$(2,2)$ is bordered double-circulant. There are {\rm(}up to
equivalence{\rm)} three such codes.
\end{prop}

\begin{proof}
We can easily determine the submodule of $B$ fixed by the given
automorphism and then do an exhaustive search with \textsc{Magma} on
its complement  in $K^{60}$ (following the methods described in
\cite{articolo18} and considering the complement as a vector space
over $\F_{2^{28}}$). In fact, it turns out that $B$ is equivalent to
one of the three bordered double-circulant singly-even codes of
length $60$ classified by Harada, Gulliver and Kaneta in
\cite{articoloHGK}.
\end{proof}

It is computationally easy to check that  there are exactly $14$
conjugacy classes of elements  of type $29$-$(2,2)$ in $\Aut(B)$ for
 each of the three possiblities for $B$.

 Using this we are able to do an exhaustive search for $C$ along the methods used in \cite{articolomio}.
 Without repeating all the details, we just recall the two main steps of the search. First we determine a set, say $\mathcal{L}$, such that there exists
 a $t\in\mathcal{S}_{120}$ and $L\in\mathcal{L}$ such that $(C(g^2)+C(g^{29}))^t=L$ and $g^t=g$. It turns out that $|\mathcal{L}|=42$. In the second step  we construct all  possible codes $C$ from the knowledge of its socle as in section VI of \cite{articolomio}. By checking the minimum distance
 we see that in all cases the codes are not extremal which proves the following.

\begin{prop}
The automorphism group of an extremal self-dual $[120,60,24]$ code
does not contain elements of order $58$.
\end{prop}

\section*{Acknowledgment}

The first author likes to express his gratitude to his supervisors
F. Dalla Volta and M. Sala. Both authors are indebted to the
\emph{Dipartimento di Matematica e Applicazioni} at Milano and the
\textit{Institut f\"ur Algebra und Geometrie} at Magdeburg for
hospitality and excellent working conditions, while this paper has
mainly been written. \emph{Laboratorio di Matematica Industriale e
Crittografia} of Trento deserves thanks for the help in the
computational part.


\begin{thebibliography}{1}

\bibitem{articolo10}
  E. F. Assmuss, H.F. Mattson, \emph{New $5$-designs},
  J. Combin. Theory 6 (1969) 122--151.

\bibitem{articolomio}
  M. Borello, \emph{The automorphism group of a self-dual $[72,36,16]$ binary code does not contain elements of order $6$}, IEEE Trans. Inform. Theory 58, No. 12 (2012), 7240--7245.

\bibitem{Baut2} S. Bouyuklieva, \emph{A method for constructing self-dual codes with an automorphism of order 2}, IEEE Trans.
Inform. Theory 46, No. 2 (2000), 496--504.

\bibitem{BMW}
 S. Bouyuklieva, A. Malevich and W. Willems, \emph{Automorphisms of extremal codes},
    IEEE Trans. Inform. Theory 56 (2010), 2091--2096.

\bibitem{articolo2}
S. Bouyuklieva, \emph{On the automorphisms of order $2$ with fixed
points for the extremal self-dual codes of length $24m$}, Des. Codes
Cryptogr. 25 (2002) 5--13.

\bibitem{articolo19}
    W. Bosma, J. Cannon, C. Playoust, \emph{The \textsc{Magma} algebra system \textsc{I}: The user
    language}, J. Symbol. Comput. 24 (1997) 235--265.

\bibitem{thesis}
J. de la Cruz, \emph{\"{U}ber die Automorphismengruppe extremaler
Codes der L\"{a}ngen 96 und 120}, PhD thesis, Otto-von-Guericke
University Magdeburg, 2012.

\bibitem{Grassl} M. Grassl, \emph{Bounds on the minimum distance of linear codes and quantum codes},
online available at {\tt www.codetables.de} , accessed on 2012-09-15

\bibitem{articoloHGK}
    M. Harada,  T.A. Gulliver and H. Kaneta, \emph{Classification of extremal double-circulant self-dual codes
of length up to $62$}, Discrete  Mathematics 188 (1998), 127--136.

\bibitem{articolo18}
    W.C.  Huffman, \emph{Automorphisms of codes with application to extremal doubly even codes of length $48$},
    IEEE Trans. Inform. Theory 28 (1982), 511--521.

\bibitem{HP}
    W.C. Huffman and V. Pless, \emph{ Fundamentals of error-correcting codes}, Cambridge University Press,
    2003.

\bibitem{HB}
   B. Huppert and N. Blackburn, \emph{Finite Groups II}, Springer
   1982.

\bibitem{MW}
   C. Mart{\'i}nez-P{\'e}rez  and W. Willems, \emph{Self-dual codes and modules of finite groups in characteristic two},
      IEEE Trans. Inform. Theory 50 (2004), 67--78.

\bibitem{articolorains}
    E.M. Rains, \emph{Shadow bounds for self-dual codes}, IEEE Trans. Inform. Theory 44 (1998), 134--139.

\bibitem{W}
   W. Willems, \emph{A note on self-dual group codes}, IEEE Trans. Inform. Theory 48 (2002), 3107--3109.


\end{thebibliography}
\end{document}